\documentclass{article}
\usepackage{a4wide}
\usepackage{latexsym}
\usepackage{amsmath}
\usepackage{amssymb}
\usepackage{amsthm}
\usepackage{ifthen}
\usepackage{mathrsfs}
\usepackage{graphics}
\usepackage{color}
\usepackage{algpseudocode}
\usepackage[ruled,vlined,linesnumbered]{algorithm2e}
\usepackage{tikz}
\usetikzlibrary{calc}
\usepackage{amsmath}
\usetikzlibrary{decorations.pathreplacing}

\usepackage{hyperref}

\usepackage[symbol]{footmisc} 
\usepackage{todonotes}

\newcommand{\nc}{\newcommand}
\nc{\rnc}{\renewcommand} \nc{\nev}{\newenvironment}

\nc{\W}[1]{\text{$\textup{W}[#1]$}}
\nc{\FPT}{\textup{FPT}}
\nc{\fpt}{\textup{fpt}}

\newlength{\probwidth}
\nc{\prob}[3][9]{
\begin{center}
  \normalfont\fbox{
   \begin{tabular}[t]{
     rp{#1cm}}\textit{Instance:}&#2. \\
     \textit{Problem:}&#3
   \end{tabular}}
\end{center}}

\nc{\pprob}[4][9]{
\begin{center}
   \normalfont\fbox{
    \begin{tabular}[t]{
     rp{#1cm}}\textit{Instance:}&#2. \\
     \textit{Parameter:}&#3. \\
     \textit{Problem:}&#4
   \end{tabular}}
\end{center}}

\nc{\nprob}[4][9]{
\begin{center}
  \normalfont\fbox{

\addtolength{\probwidth}{#1cm}\parbox{\probwidth}{\textsc{#2}\\\hspace*{1.5em}
     \begin{tabular}[t]{
      rp{#1cm}}\textit{Instance:}&#3. \\
      \textit{Problem:}&#4
     \end{tabular}}}
\end{center}}

\nc{\npprob}[5][9]{
\begin{center}
  \normalfont\fbox{

\addtolength{\probwidth}{#1cm}\parbox{\probwidth}{\textsc{#2}\\\hspace*{1.5em}
    \begin{tabular}[t]{
     rp{#1cm}}\textit{Instance:}&#3. \\
     \textit{Parameter:}&#4. \\
     \textit{Problem:}&#5
    \end{tabular}}}
\end{center}}

\nc{\nppxrob}[5][9]{ \normalfont\fbox{

\addtolength{\probwidth}{#1cm}\parbox{\probwidth}{\textsc{#2}\\\hspace*{1.5em}
   \begin{tabular}[t]{
    rp{#1cm}}\textit{Instance:}&#3. \\
    \textit{Parameter:}&#4. \\
    \textit{Problem:}&#5
   \end{tabular}}}}

\nc{\nppprob}[5][4]{
\begin{center}
  \normalfont\fbox{

\addtolength{\probwidth}{#1cm}\parbox{\probwidth}{\textsc{#2}\\\hspace*{1.5em}
    \begin{tabular}[t]{
     rp{#1cm}}\textit{Instance:}&#3. \\
     \textit{Parameter:}&#4. \\
     \textit{Problem:}&#5
    \end{tabular}}}
\end{center}}

\nc{\dotcup}{\;\dot\cup\;}

\nc{\noptprob}[6][9]{
\begin{center}
  \normalfont\fbox{

\addtolength{\probwidth}{#1cm}\parbox{\probwidth}{\textsc{#2}\\\hspace*{0em}
    \begin{tabular}[t]{
     rp{10cm}}\textit{Instance:}&#3. \\
     \textit{Solution:}&#4. \\
     \textit{Cost:}&#5. \\
     \textit{Goal:}&#6.
    \end{tabular}}}
\end{center}}

\newtheorem{theo}{Theorem}[section]
\newtheorem{theorem}[theo]{Theorem}
\newtheorem{lemma}[theo]{Lemma}

\newtheorem{definition}[theo]{Definition}
\newtheorem{remark}[theo]{Remark}

\newtheorem{claim}[theo]{Claim}

\usepackage{authblk}

\title{A Simple  Gap-producing Reduction for  the Parameterized Set Cover Problem}
\author{Bingkai Lin\\
\footnotesize National Institute of Informatics\\
\footnotesize \texttt{lin@nii.ac.jp} \\ }

\begin{document}

\maketitle
\begin{abstract}
Given an $n$-vertex bipartite graph $I=(S,U,E)$, the goal of set cover problem is to find a minimum sized subset of $S$ such that every vertex in $U$ is adjacent to some vertex of this subset.  It is NP-hard to approximate  set cover  to within a $(1-o(1))\ln n$ factor~\cite{dinste14}. If we use the size of the optimum solution  $k$ as the parameter, then it can be solved in $n^{k+o(1)}$ time~\cite{eisenbrand2004complexity}. A natural question is: can we approximate set cover to within an $o(\ln n)$ factor  in $n^{k-\epsilon}$ time?

In a recent breakthrough result~\cite{karthik2017parameterized}, Karthik, Laekhanukit and Manurangsi showed that assuming the Strong Exponential Time Hypothesis (SETH), for any computable function $f$, no $f(k)\cdot n^{k-\epsilon}$-time algorithm  can approximate set cover to a factor below $(\log n)^{\frac{1}{poly(k,e(\epsilon))}}$ for some function $e$.

This paper presents a simple gap-producing reduction which, given a set cover instance $I=(S,U,E)$  and two integers  $k<h\le (1-o(1))\sqrt[k]{\log |S|/\log\log |S|}$, outputs a new set cover instance $I'=(S,U',E')$ with $|U'|=|U|^{h^k}|S|^{O(1)}$ in $|U|^{h^k}\cdot |S|^{O(1)}$ time such that
\begin{itemize}
\item if $I$ has a $k$-sized solution, then so does $I'$;
\item if $I$ has no $k$-sized solution, then every solution of $I'$ must contain at least $h$ vertices.
\end{itemize}
Setting $h=(1-o(1))\sqrt[k]{\log |S|/\log\log |S|}$, we show that assuming SETH, for any computable function $f$, no $f(k)\cdot n^{k-\epsilon}$-time algorithm  can distinguish between  a set cover instance with $k$-sized solution and one whose minimum solution size is at least  $(1-o(1))\cdot \sqrt[k]{\frac{\log n}{\log\log n}}$. This improves  the result in~\cite{karthik2017parameterized}.

\end{abstract}

\section{Introduction}
We consider the \emph{set cover} problem  (\textsc{SetCover}): given an $n$-vertex bipartite graph $I=(S,U,E)$, where $U$ is the underlying universe set and $S$ represents the set family, find a minimum sized subset $C$ of $S$ such that every vertex of $U$ is adjacent to some vertex of $C$.  
We use $S(I)$, $U(I)$ and $opt(I)$ to denote the sets $S$, $U$ and the minimum size of the solution of $I$ respectively. A vertex $u\in U$ is \emph{covered} by a subset $C\subseteq S$ if $u$ is adjacent to some vertex of $C$. 
The set cover problem is NP-hard~\cite{kar72}. Unless $P=NP$, we do not expect to solve it in polynomial time. One way to handle NP-hard problems is to use approximation algorithms.
An algorithm of \textsc{SetCover} achieves an $r$-approximation if for every input instance $I$, it returns a subset $C$ of $S(I)$ such that $C$ covers $U(I)$ and $|C|\le r\cdot opt(I)$.
The polynomial time approximability of  \textsc{SetCover} is well-understood: the greedy algorithm can output a solution of size at most $opt(I)\cdot (1+\ln n)$~\cite{cha79,joh74,lov75,sla97,ste74} and it was shown that no polynomial time algorithm can achieve an approximation factor  within $(1-o(1))\ln n$ unless $P=NP$~\cite{alon2006algorithmic,dinste14,feige1998threshold,lunyan94,razsaf97}.
On the other hand, if we take the optimum solution size $k=opt(I)$ as a parameter, then the simple brute-force searching algorithm can solve this problem in $n^{k+1}$ time. Assuming the exponential time hypothesis (ETH)~\cite{impagliazzo2001complexity,impagliazzo2001problems}, i.e.,  $3$-SAT on $n$ variables cannot be solved in $2^{o(n)}$ time, there is no $n^{o(k)}$ time algorithm for \textsc{SetCover}. Under the strong exponential time hypothesis (SETH)~\cite{impagliazzo2001complexity,impagliazzo2001problems}, which claims that for any $\epsilon\in (0,1)$ there exists a $d\ge 3$ such that $d$-SAT on $n$ variables cannot be solved in $2^{(1-\epsilon)n}$  time, we can further rule out $n^{k-\epsilon}$-time algorithm for set cover for any $\epsilon>0$~\cite{puatracscu2010possibility}.  It is quite natural to ask~\cite{CyganFHW17}:

Is there any  $o(\ln n)$-approximation algorithm for  the parameterized set cover problem (or dominating set problem)  with running time  $n^{k-\epsilon}$?

Exponential time approximation algorithms for the unparameterised  version of set cover problem were studied in \cite{bourgeois2008efficient,cygan2009exponential}. It was shown that for any ratio $r$, there is a $(1+\ln r)$-approximation algorithm for \textsc{SetCover} with running time  $2^{n/r}n^{O(1)}$. 
No $n^{k-\epsilon}$ time algorithm  for \textsc{SetCover} achieving an approximation ratio in $o(\ln n)$ is known in literature.
On the other hand, proving inapproximability for  a parameterized  problem is not an easy task. In fact, even the  constant FPT-approximability, i.e., the existence of  $f(k)\cdot n^{O(1)}$-time algorithm for any computable function $f$ (henceforth referred to as FPT-algorithm) with constant approximation, has been open for many years~\cite{marx08}. Lacking  techniques like PCP-theorem~\cite{arora1998proof}, many results on the parameterized inapproximability of set cover problem had to use  strong conjectures~\cite{bonesc13,chalermsook2017gap}  to create a gap in the first place. 
It is of great interest to develop techniques to prove hardness of approximation for parameterized problems  only using  hypothesis such as  $SETH$, $ETH$  or even weaker assumptions like $\W 1\neq FPT$ or $\W 2\neq FPT$~\cite{dowfel99,flugro06}  from the parameterized complexity theory. The success of this quest might extend the
arsenal of methods for proving hardness of approximation and  lead  to  PCP-like theorems for Fine-Grained Complexity~\cite{abboud2017distributed}.

The first constant FPT-inapproximability result for parameterized \textsc{SetCover} based on $\W 1\neq FPT$ was given by \cite{chen2016constant} using the one-sided gap of \textsc{Biclique} from~\cite{lin15}. In fact, \cite{chen2016constant} deals with dominating set problem, which is essentially the same as \textsc{SetCover}. Recently, Karthik, Laekhanukit and Manurangsi~\cite{karthik2017parameterized} significantly improved the FPT-inapproximation factor to  $(\log n)^{1/{k^{O(1)}}}$ under the hypothesis $\W 1\neq FPT$. They also rule out the existence of $(\log n)^{1/{k^{O(1)}}}$-approximation algorithm with running time $f(k)\cdot n^{o(k)}$ for any computable function $f$,  assuming ETH, and the existence of $(\log n)^{\frac{1}{{(k+e(\epsilon))^{O(1)}}}}$-approximation algorithms  with running time $f(k)\cdot n^{k-\epsilon}$, assuming SETH. 
Their approach  is to first establish a $(\log n)^{\frac{1}{\Omega(k)}}$ gap for \textsc{MaxCover}, then reduce \textsc{MaxCover} to \textsc{SetCover} and  obtain a $(\log n)^{\frac{1}{\Omega(k^{2})}}$-gap.
This paper presents a new technique which allows us to design simple reductions improving the inapproximation factor to $(1-\epsilon)\cdot\sqrt[k]{\frac{\log n}{\log\log n}}$. The reduction in~\cite{chalermsook2017gap} can get the ratio $(\log n)^{\Omega(1/k)}$ but it has to assume Gap-ETH.

\begin{theorem}\label{thm:SETHbased}
Assuming SETH, for every $\epsilon,\delta\in(0,1)$, sufficiently large $k$\footnote{We need large $k$ to get the $\frac{1}{1+\delta}\left(\frac{\log N}{\log\log N}\right)^{\frac{1}{k}}$ gap for small $\delta$. If we want to obtain an $\Theta\left(\sqrt[k]{\frac{\log N}{\log\log N}}\right)$ gap, then our reduction works for all $k\ge 2$.} and computable function $f : \mathbb{N}\to\mathbb{N}$, there is no $f(k)\cdot N^{k-\epsilon}$ time algorithm that can, given an $N$-vertex set cover instance $I$, distinguish between
\begin{itemize}
\item $opt(I)\le k$,
\item $opt(I)>\frac{1}{1+\delta}\left(\frac{\log N}{\log\log N}\right)^{\frac{1}{k}}$.
\end{itemize}
\end{theorem}

\begin{theorem}\label{thm:ETHbased}
Assuming ETH, there is a constant $\epsilon\in(0,1)$ such that for every $\delta\in(0,1)$ and computable function $f : \mathbb{N}\to\mathbb{N}$, no $f(k)\cdot N^{\epsilon k}$ time algorithm that can, given an $N$-vertex set cover instance $I$ ,  distinguish between
\begin{itemize}
\item $opt(I)\le k$,
\item $opt(I)>\frac{1}{1+\delta}\cdot \left(\frac{\log N}{\log\log N}\right)^{\frac{1}{k}}$.
\end{itemize}
\end{theorem}

Behind these results is a reduction which, given an integer $k$, an $n$-vertex set cover instance $I$  and an integer $h\le O(\log n/\log\log n)$,  produces an $n^{O(1)}\cdot (|U(I)|)^{O(h^k)}$-vertex  instance $I'$ in $n^{O(1)}\cdot |U(I)|^{O(h^k)}$ time such that if $opt(I)\le k$ then $opt(I')\le k$, otherwise $opt(I')>h$. Therefore, to prove the $h$-factor parameterized inapproximability of \textsc{SetCover}, it suffices to show the hardness of \textsc{SetCover}  when the input instances have $n^{O(1/h^k)}$-size universe set. 
Note that the standard  reduction for SETH-hardness of set cover  parameterized by the solution size $k$ produces instances $I$ with $|U(I)|=O(k\log |S(I)|)$. With our reduction, this immediately yields the above theorems. Let us not fail to mention that the results of \cite{karthik2017parameterized} also imply the hardness of \textsc{SetCover} with logarithmic sized universe set assuming the $k$-SUM hypothesis and $\W 1\neq FPT$ hypothesis respectively.   
Similarly, we can obtain the corresponding inapproximability for set cover based on each of these hypotheses as well. In particular, using a simple trick,  we can even rule out $(\log N)^{1/\epsilon(k)}$-approximation FPT-algorithm of set cover for any unbounded computable function $\epsilon$ under $\W 1\neq FPT$.

\begin{theorem}\label{thm:kSumbased}
Assuming $k$-SUM hypothesis for any $\delta,\epsilon\in (0,1)$, sufficiently large $k$  and computable function $f : \mathbb{N}\to\mathbb{N}$,  there is no $f(k)\cdot N^{\lceil k/2\rceil-\epsilon}$ time algorithm that can, given an $N$-vertex set cover instance $I$,  distinguish between
\begin{itemize}
\item $opt(I)\le k$,
\item $opt(I)>\frac{1}{1+\delta}\left(\frac{\log N}{\log\log N}\right)^{\frac{1}{k}}$.
\end{itemize}
\end{theorem}

\begin{theorem}\label{thm:w1version}
Assuming $\W 1\neq FPT$, for  and computable function $f : \mathbb{N}\to\mathbb{N}$ and unbounded computable function $\epsilon : \mathbb{N}\to\mathbb{N}$, there is no $f(k)\cdot N^{O(1)}$-time algorithm that can, given an $N$-vertex set cover instance $I$,   distinguish between
\begin{itemize}
\item $opt(I)\le k$,
\item $opt(I)>{\log N}^{1/\epsilon(k)}$.
\end{itemize}
\end{theorem}

\subparagraph*{Technique contribution.} The main technique contribution of this paper is to introduce a  gadget that can be used to design gap-producing reductions from the set cover problem to its approximation version and provide a construction of this gadget using \emph{$(n,k)$-universal sets}. Compared to the reductions in~\cite{karthik2017parameterized}, the gap amplification step in this paper is independent of the starting assumptions. This simplifies the proof for showing the inapproximability of the set cover problem. In particular, the inapproximability result in~\cite{karthik2017parameterized} assuming SETH needs some heavy machinery like AG codes to create the gap, while our reduction is completely elementary. 

In addition to it simplicity, an important feature of our reduction is that it can  be computed by constant depth circuits. Combining this observation with Rossman's $\Omega(n^{k/4})$ size lower bound for constant depth circuits detecting $k$-clique~\cite{rossman2008constant},  Wenxin Lai~\cite{Lai19} showed that there is no constant-depth circuits of size $f(k)n^{o(\sqrt{k})}$ that can distinguish between a set cover instance with solution size at most $k$ and one whose minimum solution size is at least $({\log n}/{\log\log n})^{1/\binom{k}{2}}$.

Another advantage of our reduction is that it can give hardness approximation result from assumptions that the distributed PCP technique cannot. If we assume that $k$-set-cover  with large universe set, say $|U|=n^{1/h(k)^k}$, has no $n^{k-\epsilon}$-time algorithm, then our reduction gives $h(k)$ factor hardness of approximation $k$-set-cover in $n^{k-\epsilon}$ time. This cannot be achieved by the distributed PCP technique used in~\cite{karthik2017parameterized} due to known lower bounds in communication complexity of set disjointness.

The gap-gadget we introduce in this paper is similar to the bipartite graphs with threshold property in~\cite{lin15,Lin18}. Such kind of gadgets may have further applications in proving hardness of approximation for other parameterized problems.

\section{Preliminaries}
For $n,k\in\mathbb{N}$, an $(n,k)$-universal set is a set of binary strings with length $n$, such that the restriction to  any $k$ indices contains all the $2^k$ possible binary configurations. 
\begin{lemma}\label{lem:universalset}[See Sections 10.5 and 10.6 of \cite{jukna2011extremal}]
For $k2^k\le\sqrt{n}$,  $(n,k)$-universal sets of size $n$ can be computed in $O(n^{3})$ time.
\end{lemma}
\subparagraph*{Hypotheses.} Below is a list of hardness hypotheses we will use in this paper.
\begin{itemize}
\item $\W 1\neq FPT$: for any computable function $f: \mathbb{N}\to\mathbb{N}$,  no algorithm can, given an $n$-vertex graph $G$ and an integer $k$, decide if $G$ contains a $k$-clique in $f(k)\cdot n^{O(1)}$ time.
\item $\W 2\neq FPT$:  for any computable function $f: \mathbb{N}\to\mathbb{N}$, there is no algorithm which, given an $n$-vertex  set cover instance $I$ and an integer $k$, decides if $opt(I)\le k$ in $f(k)\cdot n^{O(1)}$ time.
\item Exponential Time Hypothesis (ETH)\cite{impagliazzo2001complexity,impagliazzo2001problems}: there exists a $\delta\in (0,1)$ such that $3$-SAT on $n$ variables cannot be solved in $O(2^{\delta n})$ time.
\item Strong Exponential Time Hypothesis (SETH)\cite{impagliazzo2001complexity,impagliazzo2001problems} for any $\epsilon\in (0,1)$ there exists $d\ge 3$ such that $d$-SAT on $n$ variables cannot be solved in $O(2^{(1-\epsilon)n})$  time.
\item $k$-SUM hypothesis ($k$-SUM) \cite{abboud2013exact}: for every  $k\ge 2$ and $\epsilon>0$, no $O(n^{\lceil k/2\rceil-\epsilon})$ time algorithm can, given  $k$ sets $S_1,\ldots,S_k$ each with $n$ integers in $[-n^{2k},n^{2k}]$, decide if there are $k$ integers $x_1\in S_1,\ldots,x_k\in S_k$ such that $\sum_{i\in[k]}x_i=0$.
\end{itemize}
We refer the reader to \cite{flugro06,dowfel99} for more information about the parameterized complexity hypotheses.
Using the Sparsification lemma~\cite{impagliazzo2001problems}, we can assume that the instances of $3$-SAT in ETH have $Cn$ clauses for some constant $C$ and the instances of $d$-SAT in SETH have $C_{d,\epsilon}n$ clauses where $C_{d,\epsilon}$ depends on $d$ and $\epsilon$.
\section{Reductions}
We start with the definition of $(k,n,m,\ell,h)$-gap-gadgets. In Lemma~\ref{lem:reduction}, we show how to use theses gadgets to create an $(h/k)$-gap for the set cover problem. Lemma~\ref{lem:constgapgadget} gives a polynomial time construction of gap-gadgets with  $h\le O({\log n/\log\log n})$ and $\ell=h^k$. Since for every input instance $I=(U,S,E)$ of set cover, our reduction runs in time $|S|^{O(1)}|U|^\ell$. If $|U|=\Omega(n)$, we can not afford such running time. Our next step is to prove the hardness of  set cover with $U=f(k)\cdot (\log n)^{O(1)}$ based on each of the aforementioned hypotheses.

\begin{definition}[$(k,n,m,\ell,h)$-Gap-Gadget]\label{def:gapgadget}
A $(k,n,m,\ell,h)$-Gap-Gadget is a bipartite graph $T=(A,B,E)$ satisfying the following conditions.
\begin{description}
\item[(G1)] $A$ is partitioned into $(A_1,A_2,\ldots,A_m)$. For every $i\in [m]$, $|A_i|=\ell$.
\item[(G2)] $B$ is partitioned into $(B_1,B_2,\ldots,B_k)$. For every $j\in[k]$, $|B_j|=n$.
\item[(G3)] For all $b_1\in B_1,b_2\in B_2,\ldots b_k\in B_k$, there exist $a_1\in A_1,\ldots,a_m\in A_m$ such that for all $i\in [m]$ and $j\in[k]$, $a_i$ is adjacent to $b_j$.
\item[(G4)] For all $X\subseteq B$ and $a_1\in A_1,\ldots,a_m\in A_m$, if every $a_i$ has at least $k+1$ neighbors in $X$, then $|X|>h$.
\end{description}
\end{definition}

To use this gadget, given a set cover instance $I=(S,U,E)$, we will identify the set $B$ with the set $S$. Then we  construct a new set cover instance $I'=(S',U',E')$ with $S'=S$ such that
\begin{itemize}
\item[($\star$)] for any subset $X$ of $S'$ that can cover $U'$, there must exist a vertex $a_i\in A_i$ for every $i\in m$  \emph{witnessing} that $X$ contains a solution of $I$, i.e., there exists $C\subseteq X$ that can cover $U$ in the instance $I$ and all the vertices of $C$ are adjacent to $a_i$ in the gap-gadget.
\end{itemize}
It is easy to check the correctness of this reduction:

If there is a $k$-vertex set $X$ that can cover $U$, then by (G3) we can pick $a_i\in A_i$ for all $i\in [m]$ such that $a_i$ is adjacent to all vertices in $X$. This means that $X$ is also a solution of $I'$. 

If $opt(I)>k$, then no matter how we pick $a_i\in A_i$, each $a_i$ must have $k+1$ neighbors in $X$. This implies that $X>h$ by (G4).

To achieve ($\star$), we will use the idea of \emph{hypercube set system} from Feige's work~\cite{feige1998threshold} (which is also used in~\cite{karthik2017parameterized,chalermsook2017gap}). For each $i\in [m]$, we construct a set $U^{A_i}$. Each element in $U^{A_i}$ can be regarded as a function $f : A_i\to U$. In the new set cover instance, $f$ is covered by $s\in S$ if there exists $a_i\in A_i$ such that $a_i$ is adjacent to $s$ in the gap-gadget and $f(a_i)$ is covered by $s$ in $I$. More details can be found in the proof of the following lemma.

\begin{lemma}\label{lem:reduction}
There is an algorithm which, given an integer $k$, an instance $I=(S,U,E)$ of \textsc{SetCover}, where $S=S_1\cup S_2\ldots \cup S_k$ and $|S_i|=n$ for all $i\in[k]$, and a $(k,n,m,\ell,h)$-Gap-Gadget, outputs a set cover instance $I'=(S',U',E')$ with $S'=S$ and $U'=m|U|^\ell$ in $|U|^\ell\cdot n^{O(1)}$ time such that
\begin{itemize}
\item if there exist $s_1\in S_1,\ldots,s_k\in S_k$ that can cover $U$, then $opt(I')\le k$;
\item if $opt(I)>k$, then $opt(I')> h$.
\end{itemize}
\end{lemma}
\begin{proof}
Let $T=(A,B,E_T)$ be the $(k,n,m,\ell,h)$-Gap-Gadget. Without loss of generality, assume that for all $i\in[k]$ $B_i=S_i$. The new instance $I'=(S',U',E')$ is defined as follows.
\begin{itemize}
\item $S'=S$.
\item $U'= (\bigcup_{i\in[m]}U^{A_i})$.
\item For all $s\in S'$ and $f\in U^{A_i}$ where $i\in [m]$, $E'$ contains $\{s,f\}$ if  there exists an $a\in A_i$ such that 
\begin{description}
\item[(E'1)] $\{s,f(a)\}\in E$, 
\item[(E'2)] $\{a,s\}\in E_T$.
\end{description}
\end{itemize}

\subparagraph*{Completeness.} If $opt(I)\le k$, then there exist $s_1\in S_1,\ldots,s_k\in S$   that can cover the whole set $U$.  We will show that for every $f\in U'$,  $f$ is covered by some vertex in $\{s_1,s_2,\cdots,s_k\}$.
Firstly, by (G3), there exist $a_1\in A_1,\ldots,a_m\in A_m$ such that $a_is_j\in E_T$ for all $i\in[m]$ and $j\in[k]$. Assume that $f\in U^{A_i}$ for some $i\in[m]$. Observe that  $f(a_i)\in U$ must be covered by some $s_j$ with $j\in [k]$, i.e., $\{s_j,f(a_i)\}\in E$. Since $\{a_i,s_j\}\in E_T$ and $\{s_j,f(a_i)\}\in E$, according to the definition of $E'$, we must have  $\{s_j,f\}\in E'$.

\subparagraph*{Soundness.} Suppose $opt(I)>k$. Let $X\subseteq S'$ be a set covering $U'$. For every $a\in A$, let $N^T(a)$ be the set of neighbors of $a$ in $T$. We have the following claim.


\begin{claim}\label{claim:Xsize}
For every $i\in[m]$ there exists $a_i\in A_i$ such that $|N^T(a_i)\cap X|\ge k+1$.
\end{claim}


\noindent\textit{Proof of Claim~\ref{claim:Xsize}.} Suppose there exists an $i\in[m]$ such that for all $a\in A_i$, $|N^T(a)\cap X|\le k$. 
Since $opt(I)>k$, every solution of $I$ has size at least $k+1$. It follows that for every $a\in A_i$, there exists some $u_a\in U$ such that $u_a$ is not covered by $N^T(a)\cap X$ in the set cover instance $I$.  Define a function $f\in U^{A_i}$ such that $f(a)=u_a$ for every $a\in A_i$. We claim that $f$ is not covered by $X$. Otherwise, suppose  there exists an  $s\in X$ that can cover $f$. According to  the definition of $E'$, there must exists an $a\in A_i$ such that   (E'1) and (E'2) hold. However, if $s\in N^T(a)\cap X$, then $\{s,f(a)\}=\{s,u_a\}\notin E$. On the other hand, if $s\notin N^T(a)\cap X$, then $\{a,s\}\notin E_T$.  In both cases, we obtain   contradictions. \flushright$\dashv$

By Claim~\ref{claim:Xsize}, we can pick $a_i\in A_i$ for each $i\in[m]$ such that every $a_i$ has at least $k+1$ neighbors in $X$. By the property of Gap-Gadget, $|X|>h$.
\end{proof}

\begin{remark}
Recall that the greedy algorithm can approximate the set cover problem within a $(1+\ln |U|)$-approximation ratio. If one could construct a gap-gadget for parameters satisfying 
\[
k(1+\ln|U'|)=  k(1+\ell\ln |U|+\ln m)<h,
\]
then applying the greedy algorithm on input $I'$ could decide whether $opt(I)=k$ in $|U|^\ell\cdot n^{O(1)}$ time. 
 
It is well known that given a CNF formula $\phi$ on $n$ variables, one can construct a set cover instance $I=(S,U,E)$ with $|U|=O(n)$ and $|S|=\Theta(k2^{n/k})$  in $2^{O(n/k)}$ time such that $\phi$ is satisfiable if and only if $opt(I)=k$. This implies that, assuming ETH there is no algorithm that can construct  $(k,|S|,m,\ell,h)$-gap-gadgets with  $k(1+\ell ln |U|+\ln m)<h$ and $|U|^\ell\le 2^{o(n)}$ in
$2^{o(n)}$ time.
\end{remark}

\subsection{Construction of Gap-Gadgets}
In~\cite{Lin18}, a similar gadget is used to prove the parameterized complexity of $k$-Biclique. One would wonder if the randomized construction from~\cite{Lin18} can be used to construct the gap-gadget in this paper.
Informally, the gadget in~\cite{Lin18} is a bipartite random graph $T=(A,B,E)$ satisfying  the following properties with high probability:
\begin{description}
\item[(T1)] a $k$-vertex set in $B$ has $m=n^{\Theta(1/k)}$ common neighbors;
\item[(T2)] any $(k+1)$-vertex set in $B$ has at most $O(k^2)$ common neighbors.
\end{description}
 It is not hard to show that if $Y\subseteq A$ is an $m$-vertex set and every vertex in $Y$ has at least $k+1$ neighbors in $X\subseteq B$, then $|X|\ge \sqrt[k+1]{\frac{|Y|}{O(k^2)}}$ by (T2) and  the pigeonhole principle. We may partition the vertex set $A$ into $m$ parts. Each part contains $n^{1-\Theta(1/k)}$ vertices. This gives us  a gap-gadget with large gap $h=\sqrt[k+1]{\frac{m}{O(k^2)}}$ and $\ell=n^{1-\Theta(1/k)}$. Unfortunately, such gadget does not suit our purpose. We need a gap-gadget with $\ell\le \log n/\log\log n$. In this section, we provide a construction using universal sets.

\begin{lemma}\label{lem:constgapgadget}
There is an  algorithm that can, for every $k,h,n\in\mathbb{N}$ with $k\log\log n\le \log n$ and $h\le \frac{\log n}{(2+\epsilon)\log\log n}$, compute a $(k,n,n\log h,h^k,h)$-Gap-Gadget   in $O(n^{4})$ time.
\end{lemma}

\begin{proof}
Let $m=n\log h$ and $K=h\log h$.
Note that $(\log m)/2=(\log n+\log\log h)/2\ge (2+\epsilon) h\log h/2\ge \log h+\log\log h+h\log h=\log K+K$, i.e., $K2^K\le \sqrt{m}$.
By Lemma~\ref{lem:universalset},   an $(m,K)$-universal set $S=\{s_1,s_2,\ldots,s_m\}$ can be constructed in $O(m^{3})\le O(n^4)$ time. Partition every $s\in S$ into $n=\frac{m}{\log h}$ blocks so that each block has length $\log h$. Interpret the values of  blocks as integers in $[h]$.  We obtain an $m\times n$ matrix $M$ by setting the value  $M_{r,c}$  equal to the value of the $c$-th block of $s_r$. The matrix $M$  satisfies the following conditions.
\begin{description}
\item[(M1)] For all $r\in[m]$ and $c\in[n]$, $M_{r,c}\in[h]$.
\item[(M2)] For any set $C\subseteq [n]$ with $|C|\le h$, there exists a row $r\in [m]$ such that $|\{M_{r,c} : c\in C\}|=|C|$.
\end{description}
Condition (M1) is obvious. To see why (M2) holds, for each $C\subseteq  [n]$ with $|C|\le h$, let $C'$ be the set of indices corresponding to the blocks in $C$. Note that $|C'|=|C|\log h\le h\log h=K$. By the property of $(m,K)$-universal set, there exists an $s_r\in S$ such that each block in $C$ takes distinct value. It follows that $|\{M_{r,c} : c\in C\}|=|C|$.

\medskip

For each $i\in [m]$, let 
\[
A_{i}=\{(a_1,a_2,\ldots,a_k) : \text{for all $j\in[k]$, $a_j\in [h]$}\}.
\] 
Note that $|A_i|=h^k$. For each $j\in[k]$, let $B_j=[n]$. Let $T=(A,B,E)$ be a bipartite graph with
\begin{itemize}
\item $A=\bigcup_{i\in[m]}A_{i}$.
\item $B=\bigcup_{j\in[k]}B_j$.
\item $E=\{\{\vec{a},b\} : \text{$\vec{a}\in A_{i}, b\in B_j$ and  $M_{i,b}=\vec{a}[j]$  for all $j\in [k]$}\}$.
\end{itemize}
We will show that $T$ is an $(k,n,m,h^k,h)$-gap-gadget. Obviously, $T$ satisfies (G1) and (G2).

\subparagraph*{$T$ satisfies (G3).} For any $b_1\in B_1,b_2\in B_2,\ldots,b_k\in B_k$. We define $\vec{a}_{i}\in A_{i}$ by setting \[
\vec{a}_{i}=(M_{i,b_1},M_{i,b_2},\ldots,M_{i,b_k}).
\]  
It is routine to check that $\{\vec{a}_{i}, b_j\}\in E$ for all $i\in[m]$ and $j\in[k]$.

\subparagraph*{$T$ satisfies (G4).} Let $X\subseteq B$ and $\vec{a}_1\in A_1, \vec{a}_2\in A_2,\ldots, \vec{a}_{m}\in A_{m}$.  
Suppose for every $i\in[m]$, $\vec{a}_{i}$ has at least $k+1$ neighbors in $X$ and $|X|\le h$.
By (M2), there exists an $r\in [m]$ such that $|\{M_{r,c}: c\in X\}|=|X|$. Since $\vec{a}_r$ has at least $k+1$ neighbors in $X$, there exists an $j\in [k]$ such that $\vec{a}_r$ has two neighbors $b,b'$ in $X\cap B_j$. According to the definition of $E$, we must have
\[
M_{r,b}=M_{r,b'}=\vec{a}_r[j].
\]
This contradicts the fact that $|\{M_{r,c}: c\in X\}|=|X|$.
\end{proof}

The  construction above produces gap-gadgets with $\ell=h^k$. Note that the parameter  $h$ is related to the inapproximation factor we will get for the set cover problem and  the running time of our reduction is $n^{O(1)}|U|^\ell$. We want to set $h$ as large as possible while keeping the running time of reduction  in $f(k)\cdot n^{O(1)}$. Assuming $|U|=g(k)\cdot (\log n)^{O(1)}$, the best we can achieve is $h= (\log n/\log\log n)^{1/k}$.

\subparagraph*{On the probabilistic construction.} A natural question is, can we construct gap-gadgets with better parameters $h$ and $\ell$, say $\ell=h=o(\log n)$,  using the probabilistic method? 

Consider the probability space  of bipartite random graphs on the vertex sets $A=A_1\cup A_2\cup\cdots\cup A_m$ and $B=B_1\cup B_2\cup\cdots B_k$, where $|A_i|=\ell$ and $|B_j|=n$. Let $p$ be the edge probability. Each bipartite graph $T$ on $A\cup B$ has probability $Pr[T]=p^{|E(T)|}(1-p)^{|A|\cdot|B|-|E(T)|}$. Fix  $k$ vertices $b_1,b_2,\ldots,b_k$ in $B$. Let $X_{good}$ be the random variable  that for every bipartite graph $T$, $X_{good}(T)$ is the number of complete bipartite subgraphs of $T$ which contains  exactly one vertex in each $A_i$ and the $k$ vertices $b_1,b_2,\ldots,b_k$ in $B$. Let $X_{bad}$ be the random variable  that  for every bipartite graph $T$, $X_{bad}(T)$ is the number of  subgraphs of $T$ with $h$ vertices in $B$ and one vertex in each $A_i$ such that each vertex in $A_i$ has at least $k+1$ neighbors in $B$. We want to set the edge probability $p$ so that $\Pr[X_{bad}(T)\ge 1]+\Pr[X_{good}(T)=0]\le 1-n^{-c}$ for some constant $c>0$.  One way to bound $\Pr[X_{bad}(T)\ge 1]$ above is to use  Markov's inequality, which gives us $\Pr[X_{bad}(T)\ge 1]\le E[X_{bad}]$. So we might assume that $E[X_{bad}]<1$. On the other hand, we have $E[X_{good}]\ge \Pr[X_{good}(T)\ge 1]\ge n^{-c}$. Note that expectations of these two random variables are  $E[X_{good}]=\ell^{m}p^{km}$  and $E[X_{bad}]=\ell^m \binom{n}{h} p^{(k+1)m}\binom{h}{k+1}^m$. We deduce that 
\[
m\log\ell+mk\log p>-c\log n
\] 
and 
\[ 
m\log\ell+h\log n+m(k+1)\log p+m(k+1)\log h<0. 
\] 
Thus 
\begin{equation}\label{eq:hell}
\frac{c\log n}{mk}+\frac{\log \ell}{k}>\frac{\log\ell}{(k+1)}+\frac{h\log n}{m(k+1)}+\log h.
\end{equation} 
We might choose $m$ large enough so that  the terms $\frac{c\log n}{mk}$ and $\frac{h\log n}{m(k+1)}$ in (\ref{eq:hell}) become  relatively small. In order to make (\ref{eq:hell}) hold, we have to set  $\ell\ge h^{O(k^2)}$. This does not give us better $(k,n,m,\ell,h)$-gap-gadgets.

\subsection{Proofs of  Theorem~\ref{thm:SETHbased} and Theorem~\ref{thm:ETHbased}}\label{sec:all}

\begin{lemma}\label{lem:sat2setcover}
There is an algorithm, which given  $k\in\mathbb{N}$, $\delta>0$ with $(1+1/k^3)^{1/k}\le (1+\delta)/(1+\delta/2)$ and $(1+\delta/2)^k\ge 2k^4$   and a SAT instance $\phi$ with $n$ variables and $Cn$ clauses, where $n$ is much larger than $k$ and $C$,
outputs an integer $N\le 2^{n/k+n/k^3}$ and a set cover instance $I$  satisfying the following conditions in $2^{5n/k}$ time.
\begin{itemize}
\item $|S(I)|+|U(I)|\le N $. 
\item If $\phi$ is satisfiable, then $opt(I)\le k$.
\item If $\phi$ is not satisfiable, then $opt(I)> \frac{1}{1+\delta}\cdot\sqrt[k]{\frac{\log N}{\log\log N}}$.
\end{itemize}
\end{lemma}
\begin{proof}
Let $k$ be a positive integer and $\phi$ be a CNF with $n$ variables and $Cn$ clauses. We first construct a set cover instance $I'=(S',U',E')$  as follows.  Partition the variable set into $k$ parts, each  having at most $\lceil n/k\rceil$ variables. For each $i\in [k]$, let $S_i$ be the set of assignments to the $i$-th part. Let $S'=S_1\cup\cdots\cup S_k$. Let $U'$ be the set consisting of all the clauses of $\phi$ and $k$ additional nodes $u_1,u_2,\ldots,u_k$. For every $i\in[k]$ and assignment $s\in S_i$, we add an edge between $s$ and $u_i$.  If the assignment $s\in S'$ satisfies a clause $u\in U'$, we also add an edge between  $u$ and  $s$. The set cover instance $I'$ has the following properties.
\begin{itemize}
\item If $\phi$ is satisfiable, then $opt(I')=k$. Moreover, there exist $k$ vertices $s_1\in S_1,\cdots,s_k\in S_k$ that can cover the whole set $U'$.
\item If $\phi$ is not satisfiable, then $opt(I')>k$.
\item $|U'|=k+Cn$.
\item $|S'|\le k2^{n/k}$.
\end{itemize}
Let $M= k2^{n/k}\ge |S|$ and $N=M^{1+1/k^3}\le 2^{n/k+n/k^3}$. Note that $\log M/\log\log M\ge n/(k\log n)\ge k$. 
Applying Lemma~\ref{lem:constgapgadget} with $k\gets k$, $n\gets M$,  $\ell\gets\frac{\log M}{(1+\delta/2)^k\log\log M}$, $h\gets\frac{1}{1+\delta/2}\cdot\sqrt[k]{\frac{\log M}{\log\log M}}$ and $m\gets M\log h\le M\log\log M$, we obtain a gap-gadget $T$ in $O(M^{4}\le 2^{5n/k})$ time. Using Lemma~\ref{lem:reduction} on $I'$ and $T$, we obtain our target set cover instance $I=(S,U,E)$ satisfying the following properties.
\begin{itemize}
\item If $\phi$ is a yes-instance, then $opt(I)\le k$.
\item If $\phi$ is a no-instance, then $opt(I)>\frac{1}{1+\delta/2}\cdot\sqrt[k]{\log M/\log\log M}$. Using $(1+1/k^3)^{1/k}\le (1+\delta)/(1+\delta/2)$,  we get $opt(I)> \frac{1}{1+\delta}\cdot \sqrt[k]{\log N/\log\log N}$.
\item $|S|=|S|\le k2^{n/k}$.
\item $|U|\le M\log\log M\cdot |U|^{\frac{\log M}{(1+\delta/2)^k\log\log M}}=M\log\log M \cdot (k+Cn)^{\frac{\log M}{(1+\delta/2)^k\log\log M}}$.
\end{itemize}
The number of vertices in $I$ is 
\begin{align*}
|S(I)|+|U(I)|&\le M+M\log\log M\cdot (k+Cn)^{\frac{\log M}{(1+\delta/2)^k\log\log M}}\\
&\le M+M\log\log M\cdot (2Ck\log M)^{\frac{\log M}{(1+\delta/2)^k\log\log M}}\\
&\le M+M\log\log M\cdot (\log M)^{\frac{2\log\log M}{(1+\delta/2)^k\log\log M}}\quad\text{(using $ \log M\ge 2Ck$ for large $n$)}\\
&\le M+M\log\log M\cdot M^{\frac{2}{(1+\delta/2)^k}}\\
&\le M+M\log\log M\cdot  M^{1/k^4}\quad\text{(using $(1+\delta/2)^k\ge 2k^4$)}\\
&\le M^{1+1/k^3}\quad\text{(using $M^{1/k^3}\ge 1+ M^{1/k^4}\log\log M$ for large $n$)}\\
&= N.
\end{align*}
\end{proof}

Now we are ready to prove  Theorem~\ref{thm:SETHbased}. 
Suppose for some computable function $f$,  there is an $f(k)\cdot N^{k-\epsilon}$-time algorithm that can, for every  $N$-vertex set cover instance $I$ and every integer $k$, distinguish between $opt(I)\le k$ and $opt(I)\ge\frac{1}{1+\delta}\cdot \sqrt[k]{\frac{\log N}{\log\log N}}$. For every $\delta\in(0,1)$, choose $k\in\mathbb{N}$ large enough so that  $(1+1/k^3)^{1/k}\le (1+\delta)/(1+\delta/2)$ and $(1+\delta/2)^k\ge 2k^4$ hold. Let $\epsilon'=1-\epsilon/k+1/k^2$, by SETH, there exists an integer $d$ such that  $d$-SAT with $n$ variables cannot be solved in $2^{n(1-\epsilon')}$-time. Given an instance $\phi$ of $d$-SAT with $n$ variables and $m$ clauses. By the sparsification lemma~\cite{impagliazzo2001problems}, we can assume that   $m=C_{d,\epsilon'}\cdot n$ for some constant $C_{d,\epsilon'}$ depending on $d$ and $\epsilon'$. Without loss of generality, assume that $n$ is much larger than  $k$. Applying Lemma~\ref{lem:sat2setcover} on $\phi$ and $k$, we obtain a set cover instance $I$ with $N\le 2^{n/k+n/k^3}$ vertices in time $2^{5n/k}\le 2^{\epsilon n}$ for $k\ge 5/\epsilon$.  Then we use the approximation algorithm to decide if $opt(I)\le k$ or $opt(I)\ge\frac{1}{1+\delta}\cdot \sqrt[k]{\frac{\log N}{\log\log N}}$.  Thus we can solve $d$-SAT in time $2^{\epsilon n }+ f(k)\cdot N^{k-\epsilon}\le 2^{\epsilon n}+f(k)\cdot 2^{(n/k+n/k^3)(k-\epsilon)}\le 2^{n(1-\epsilon/k+1/k^2)}=2^{n(1-\epsilon')}$, which contradicts SETH.

Theorem~\ref{thm:ETHbased} can be proved  similarly. By ETH, there exists $\epsilon>0$ such that $3$-SAT on $n$ variables cannot be solved in $2^{\epsilon n}$ time. Let $\epsilon'=\epsilon/2$. For every  $3$-SAT instance $\phi$ with $n$ variable and $Cn$ clause, where $n$ is much larger than $k$, apply Lemma~\ref{lem:sat2setcover} to obtain a set cover instance $I$ with $N=2^{n/k+n/k^3}$ vertices in $2^{5n/k}\le 2^{\epsilon' n}$ time. If there is an $f(k)\cdot N^{\epsilon'k}$-time algorithm that can distinguish between $opt(I)\le k$ and $opt(I)>\frac{1}{1+\delta}\cdot \sqrt[k]{\frac{\log N}{\log\log N}}$, then we can decide whether $\phi$ is satisfiable in time $2^{\epsilon'n}+f(k)\cdot 2^{(n/k+n/k^3)\cdot \epsilon'k}\le 2^{\epsilon n}$.

\subsection{Proof of Theorem~\ref{thm:kSumbased}}
We use a lemma in~\cite{abboud2014losing} to reduce  $k$-SUM to $k$-VECTOR-SUM over small numbers. Then we present a reduction from $k$-VECTOR-SUM to set cover.
\begin{lemma}[Lemma 3.1 of \cite{abboud2014losing}]\label{lem:ksum2kvectorsum}
Let $k,p,d,s,M\in\mathbb{N}$ satisfy $k<p$, $p^d\ge kM+1$, and $s=(k+1)^{d-1}$. There is a collection of mappings $f_1,\ldots,f_s : [0,M]\times [0,kM]\to [-kp,kp]^d$, each computable in time $O(poly \log M+k^d)$, such that for all numbers $x_1,\ldots,x_k\in [0,M]$ and targets $t\in [0,kM]$,
\[
\sum_{j=1}^kx_j=t \Leftrightarrow \exists i\in[s] \text{ such that } \sum_{j=1}^k f_i(x_j,t)=\vec{0}.
\]
\end{lemma}

\begin{lemma}\label{lem:ksum2setcover}
There is an algorithm which, given $k$ sets $S_1,S_2,\ldots,S_k$ where $S_i$ is a set of $n$ vectors in $[-f(k),f(k)]^{g(k)\log n}$ for some computable functions $f$ and $g$,  outputs a set cover instance  $I=(S,U,E)$ with $|U|\le k^{(2f(k))^{k-1}}g(k)\log n$ and $S=S_1\cup S_2\cup\ldots\cup S_k$ in $k^{(2f(k))^{k-1}}g(k)n^{O(1)}$-time such that
\begin{itemize}
\item[(i)] if there exist $\vec{x}_1\in S_1,\ldots,\vec{x}_k\in S_k$ such that $\sum_{i\in [k]}\vec{x}_i=\vec{0}$, then $\{\vec{x}_1,\ldots,\vec{x}_k\}$ covers $U$;
\item[(ii)] if  the sum of any $k$ vectors  $\vec{x}_1\in S_1,\ldots \vec{x}_k\in S_k$ is not zero, then $opt(I)>k$.
\end{itemize}
\end{lemma}
\begin{proof}
Let $D=\{(d_1,\ldots,d_k)\in[-f(k),f(k)]^k : \sum_{i\in[k]}d_i=0\}$. Note that $|D|\le (2f(k))^{k-1}$. Suppose $D=\{\vec{a}_1,\ldots,\vec{a}_{|D|}\}$. For every $j\in [g(k)\log n]$, let $U_j=[k]^{|D|}$. We define the target set cover instance $I=(S,U,E)$ as follows.
\begin{itemize}
\item $S=S_1\cup\cdots\cup S_k$.
\item $U=\bigcup_{i\in[g(k)\log n]}U_i$.
\item For every $\vec{x}\in S_i$ and every $\vec{u}\in U_j$, we add an edge $\{\vec{x},\vec{u}\}$ into $E$ if there exists $\ell\in[|D|]$ such that  $\vec{u}[\ell]=i$ and $\vec{x}[j]=\vec{a}_\ell[i]$.
\end{itemize}

\subparagraph*{Completeness.} Suppose there exist $\vec{x}_1\in S_1,\ldots,\vec{x}_k\in S_k$ such that $\sum_{i\in[k]}\vec{x}_i=\vec{0}$. Then for all $j\in [g(k)\log n]$ we have $\vec{x}_1[j]+\vec{x}_2[j]+\ldots+\vec{x}_k[j]=0$, i.e.,
\begin{equation}\label{eq:aell}
(\vec{x}_1[j],\vec{x}_2[j],\ldots,\vec{x}_k[j])=\vec{a}_\ell\in D\text{ for some $\ell\in[|D|]$}.
\end{equation}
For all $\vec{u}\in U_j$, let $i=\vec{u}[\ell]\in [k]$. Then by (\ref{eq:aell}), $\vec{x}_i[j]=\vec{a}_\ell[i]$. It follows that $\{\vec{x}_i,\vec{u}\}\in E$.
\subparagraph*{Soundness.} Suppose the sum of any $k$ vectors in $S_1\cup \cdots\cup S_k$ is not zero. Let $X$ be a subset of $S$ with $|X|\le k$, we need to show that $X$ does not cover $U$. Firstly, we note that if  $X\cap S_i=\emptyset$ for some $i\in[k]$, then the vector $\vec{u}=(i,i,\ldots,i)\in [k]^{|D|}$ is not covered by any vector in $X$. Now assume that $X=\{\vec{x}_1,\vec{x}_2,\ldots,\vec{x}_k\}$ and $\vec{x}_i\in S_i$ for all $i\in[k]$. Since $\sum_{i\in[k]}\vec{x}_i\neq\vec{0}$, there exists a $j\in[g(k)\log n]$ such that
\[
\sum_{i\in [k]}\vec{x}_i[j]\neq 0.
\] 
 We deduce that
\[
(\vec{x}_1[j],\vec{x}_2[j],\ldots,\vec{x}_k[j])\notin D.
\]
In other word, for all $\ell\in[|D|]$, there exists an $i_\ell\in [k]$ such that
\begin{equation}\label{eq:xilneqal}
\vec{x}_{i_\ell}[j]\neq \vec{a}_\ell[i_\ell].
\end{equation}
Define a vector $\vec{u}\in U_j$ such that for all $\ell\in [|D|]$,
\begin{equation}\label{eq:defu}
\vec{u}[\ell]=i_\ell.
\end{equation}
Suppose $\vec{u}$ is covered by $x_i\in X$, then by the definition, there exists $\ell\in[|D|]$ such that $i=\vec{u}[\ell]=i_\ell$  and $\vec{x}_{i_\ell}[j]=\vec{a}_\ell[i_\ell]$, which contradicts (\ref{eq:xilneqal}) and (\ref{eq:defu}).
\end{proof}

\noindent\textit{Proof of Theorem~\ref{thm:kSumbased}.} Given $k$ sets $S_1,\ldots,S_k$ of integers in $[-n^{2k},n^{2k}]$. 
Let $p=k^{4k^{c+1}}$, $M=2n^{2k}$ and $d=\log n/k^c$. Without loss of generality, assume that $k$ is large and  $n$ is much larger than $k$, we have $p^d=k^{4k\log n}\ge n^{4k}\ge 2kn^{2k}+1$. On the other hand, for any $\epsilon>0$, we can pick $c$ such that $s=(k+1)^d=n^{\log(k+1)/k^c}\le n^{\epsilon/4}$. Applying Lemma~\ref{lem:ksum2kvectorsum}, we obtain a collection of mappings $f_1,\ldots,f_s : [0,M]\times [0,kM]\to [-kp,kp]^d$ in $O(poly\log M+k^d)$ time such that
\begin{itemize}
\item  there exist $x_1\in S_1,\ldots,x_k\in S_k$ with $\sum_{j\in[k]}x_j=0$ if and only if there exist $i\in[s]$ such that $\sum_{j\in[k]}f_i(x_j+n^{2k},kn^{2k})=\vec{0}$.
\end{itemize}
Using Lemma~\ref{lem:constgapgadget}, we construct a $(k,n,O(n\log\log n), \frac{\log n}{(1+\delta/2)^k\log\log n},\frac{1}{(1+\delta/2)}\cdot(\frac{\log n}{\log\log n})^{1/k})$-gap-gadget $T$ for some small $\delta>0$.
For every $i\in[s]$, and $j\in[k]$, let $S_j^i=\{f_i(x+n^{2k},kn^{2k}) : x\in S_j\}$. Applying Lemma~\ref{lem:reduction} to  $S_1^i,S_2^i,\ldots,S_k^i$ and $T$, we obtain a set cover instance $I_i$ with $S(I_i)=S_1^i\cup S_2^i\ldots S_k^i$ and $|U(I_i)|\le n\log\log n\cdot (g(k)\log n)^{\frac{\log n}{(1+\delta/2)^k\log\log n}}\le n^{1+1/k^3}$. The set cover instances $I_1,\ldots,I_s$ satisfy the following properties.
\begin{itemize}
\item If there exist $x_1\in S_1,\ldots,x_k\in S_k$ with $\sum_{j\in[k]}x_j=0$, then there exist $i\in[s]$ and $y_1=f_i(x_1+n^{2k},n^{2k})\in S_1^i\ldots y_k=f_i(x_k+n^{2k},n^{2k})\in S_k^i$ such that $y_1,\ldots,y_k$ cover $U(I_i)$.
\item If there are no $x_1\in S_1,\ldots,x_k\in S_k$ with $\sum_{j\in[k]}x_j=0$, then for all $i\in[s]$, $opt(I_i)> \frac{1}{1+\delta/2}\cdot \left(\frac{\log n}{\log\log n}\right)^{1/k}$.
\end{itemize}
Let $N=n^{1+1/k^2}$. We have  
\[
|S(I_i)|+|U(I_i)|\le kn+ n^{1+1/k^3}\le N,
\] 
\[
f(k)\cdot N^{\lceil k/2\rceil-\epsilon}\le n^{\lceil k/2\rceil-\epsilon+1/k},
\] 
and 
\[
\frac{1}{(1+\delta)}\left(\frac{\log N}{\log\log N}\right)^{1/k}\le \frac{1}{(1+\delta/2)}\left(\frac{\log n}{\log\log n}\right)^{1/k}.
\]
For every $i\in [s]$, we apply the  $f(k)\cdot N^{\lceil k/2\rceil-\epsilon}$-time algorithm  to decide if $opt(I_i)\le k$ or $opt(I_i)>\frac{1}{1+\delta}\cdot (\log N/\log\log N)^{1/k}$. If for some $i\in[s]$, it found that $opt(I_i)\le  k$, then we know that the input instance of $k$-SUM is a yes-instance. The running time is $O(poly\log M+k^d)+f(k)\cdot N^{\lceil k/2\rceil-\epsilon}\le O(poly\log M+k^d)+s\cdot n^{\lceil k/2\rceil-\epsilon+1/k}\le n^{\lceil k/2\rceil-\epsilon/2}$ for large $k$. 
\subsection{Proof of Theorem~\ref{thm:w1version}}\label{sec:logu}
Firstly, we give a reduction from \textsc{Clique} to \textsc{Set-Cover} which produces instances with logarithmic sized universe set. The main idea of this reduction is due to Karthik et al.~\cite{karthik2017parameterized}.
\begin{lemma}\label{lem:clique2setcover}
There is an $n^{O(1)}$-time  algorithm which, given an integer $k$, an $n$-vertex  graph $G$ with $V(G)=V_1\cup V_2\cup\cdots\cup V_k$ such that  $G[V_i]$ is an independent set for all $i\in[k]$,  outputs a set cover  instance $I=(S,U,E)$ with $|U|=k^{O(1)}\log n$ and $S=E(G)=\bigcup_{\{i,j\}\in\binom{[k]}{2}} S_{\{i,j\}}$, where each $S_{\{i,j\}}$ is the set of edges between $V_i$ and $V_j$,  such that
\begin{itemize}
\item[(i)] if $G$ contains a $k$-clique,  then $opt(I)\le\binom{k}{2}$. Moreover, there exists a $\binom{k}{2}$-sized subset of $S$, which contains exactly one vertex from each $S_{\{i,j\}}$ ($\{i,j\}\in\binom{[k]}{2}$),  that can cover $U$;
\item[(ii)] if $G$ contains no $k$-clique, then $opt(I)>\binom{k}{2}
$.
\end{itemize} 
\end{lemma}
\begin{proof}
We will construct a set cover instance $I$ such that if $G$ has a $k$-clique, then we can select its $\binom{k}{2}$ edges to cover the whole universe set. 
For every $v\in V(G)$,  denote by $encode(v)\in\{0,1\}^{\log n}$ the binary string representation of $v$. For every $\ell\in[\log n]$, the $\ell$th bit of $encode(v)$ is  $encode(v)[\ell]$. 
For every $i\in [k]$, let $\sigma_i : [k]\setminus \{i\}\to [k-1]$ be an arbitrary bijection. Our target set cover instance $I=(S,U,E)$ is defined  as follows.
\begin{itemize}
\item $S=E(G)=\bigcup_{\{i,j\}\in\binom{[k]}{2}}S_{\{i,j\}}$, where $S_{\{i,j\}}=\{\{v_i,v_j\} : v_i\in V_i,v_j\in V_j,\{v_i,v_j\}\in E(G)\}$.
\item $U=[k]\times [k-1]^{\{0,1\}}\times [\log n]$.
\item For $s=\{v_i,v_j\}\in S$ and $u=(i,f,\ell)\in U$ we add $\{s,u\}$ into $E$ if
\[
\text{$v_i\in V_i$, $v_j\in V_j$ and $f(encode(v_i)[\ell])=\sigma_{i}(j)$.}
\]
\end{itemize}
The set cover instance $I$ satisfies the following conditions.
\begin{itemize}
\item If $G$ contains a $k$-clique, then there exists a $\binom{k}{2}$-sized subset of $S$ which contains exactly one vertex from each $S_{\{i,j\}}$ ($\{i,j\}\in\binom{[k]}{2}$)  that can cover $U$. 
Suppose that $v_1\in V_1,\ldots,v_k\in V_k$ induce a $k$-clique. Let $X=\{\{v_i,v_j\}:\{i,j\}\in \binom{[k]}{2}\}$. We will show that $X$ covers the whole set $U$. For any $(i,f,\ell)\in U$, let $b=encode(v_i)[\ell]$. Since $f(b)\in [k-1]$, there must exist a $j\in [k]\setminus \{i\}$ such that $\sigma_i(j)=f(b)$.
By the definition of $E$, $\{v_i,v_j\}$ is adjacent to $(i,f,\ell)$ .
\item If $G$ does not contain a $k$-clique, then $opt(I)>\binom{k}{2}$. Let $X\subseteq S$ be a set such that $|X|\le \binom{k}{2}$ and $X$ covers $U$.  

For each $\{i,j\}\in\binom{[k]}{2}$, define
\[
X_{\{i,j\}}=\{\{v_i,v_j\} : v_i\in V_i, v_j\in V_j, \{v_i,v_j\}\in X\}.
\]
We claim that for every $\{i,j\}\in\binom{[k]}{2}$, $|X_{\{i,j\}}|>0$. Otherwise let $f(0)=f(1)=\sigma_i(j)$ and consider the vertex  $(i,f,1)\in U$. According to the definition of $E$, if a vertex $\{v,u\}\in S$ covers $(i,f,1)$, then either $v$ or $u$ must be in $V_i$. Let us assume $v\in V_i$ and $u\in V_{j'}$ for some $j'\in [k]\setminus \{i\}$. We must have $f(encode(v_i)[1])=\sigma_i(j')$. However, if $j\neq j'$, then $f(0)=f(1)=\sigma_i(j)\neq \sigma_i(j')$.

Since $\binom{k}{2}\ge |X|=\sum_{\{i,j\}\in \binom{[k]}{2}}|X_{\{i,j\}}|$ and $|X_{\{i,j\}}|>0$, we conclude that $|X_{\{i,j\}}|=1$ for all $\{i,j\}\in\binom{[k]}{2}$.

For every $i\in [k]$ and distinct $j,j'\in [k]\setminus\{i\}$, let $\{\{v,v_j\}\}= X_{\{i,j\}}$ and $\{\{v',v_{j'}\}\}= X_{i,j'}$, where $v,v'\in V_i$, we claim that $v=v'$. Otherwise, since $v\neq v'$ there exists $\ell\in[\log n]$ such that $encode(v)[\ell]\neq encode(v')[\ell]$. Now consider a function $f$ with $f(encode(v')[\ell])=\sigma_i(j)$ and $f(encode(v)[\ell])=\sigma_i(j')$. The vertex $(i,f,\ell)$ must be covered by some $\{x,y\}$ with $x\in V_i$ and $y\in V_h$ such that $\sigma_i(h)=f(encode(v)[\ell])\in\{\sigma_i(j),\sigma_i(j')\}$. We must have $y\in V_j$ or $y\in V_{j'}$. Since $|X_{\{i,j\}}|=|X_{\{i,j'\}}|=1$, we deduce that either $\{x,y\}=\{v,v_j\}$ or $\{x,y\}=\{v',v_{j'}\}$. However, if $\{x,y\}=\{v,v_j\}$, we must have $\sigma_i(j)=f(encode(v)[\ell])=\sigma_i(j')\neq \sigma_i(j)$, a contradiction. Similarly, if $\{x,y\}=\{v',v_{j'}\}$, then  $\sigma_i(j')=f(encode(v')[\ell])=\sigma_i(j)\neq \sigma_i(j')$. We conclude that the vertex $(i,f,\ell)$ can not be covered by $X$.

Now we have for every $i\in[k]$, there exists a $v_i\in V_i$ such that 
\[
\{v_i\}=\bigcap_{j\in [k]\setminus \{i\}, e\in X_{\{i,j\}}}e.
\] 
Obviously,  for every $\{i,j\}\in \binom{[k]}{2}$, $\{\{v_i,v_j\}\}= X_{\{i,j\}}$. This implies that $\{v_1,v_2,\ldots,v_k\}$ is a $k$-clique in $G$.
\end{itemize}
\end{proof}

\noindent\textit{Proof of Theorem~\ref{thm:w1version}.} Given an $n$-vertex graph $G$ and a positive integer $k$, we invoke Lemma~\ref{lem:clique2setcover} to obtain a set cover instance $I=(S,U,E)$ with $|S|=|E(G)|$ and $|U|\le k^3\log n$ satisfying (i) and (ii). Let $m=|S|$. Then we use  Lemma~\ref{lem:constgapgadget} to construct a $(\binom{k}{2},m,n^{O(1)},\frac{\log m}{\log\log m},\frac{\log m}{\log\log m}^{1/\binom{k}{2}})$-gap-gadget $T$ in $m^{O(1)}=n^{O(1)}$ time. Applying Lemma~\ref{lem:reduction} on $I$ and $T$, we finally obtain our target set cover instance $I'=(S',U',E')$ with the following properties:
\begin{itemize}
\item if $G$ has a $k$-clique, then $opt(I')=\binom{k}{2}$,
\item if $G$ has no $k$-clique, then $opt(I')>\left(\frac{\log m}{\log\log m}\right)^{1/\binom{k}{2}}$,
\item  $|S'|=|E(G)|=m$,
\item $|U'|=(k^3\log n)^{\log m/\log\log m}=m^{1+o(1)}$.
\end{itemize}
Let $N=|U'|+|S'|$. We have $N=n^{O(1)}$.  Since $\epsilon$ is an unbounded computable function, there is  a computable function  $g: \mathbb{N}\to\mathbb{N}$ such that $k'=g(k)>\binom{k}{2}$ and $\epsilon(k')>\binom{k}{2}$.  When $n$ is large enough,  
\[
\frac{\log m}{\log\log m}^{1/\binom{k}{2}}\ge\frac{\log N}{O(\log\log N)}^{1/\binom{k}{2}}\ge (\log N)^{1/\epsilon(k')}.
\]  
Any $f(k')\cdot N^{O(1)}$ time algorithm that can distinguish between $opt(I')\le k'$ and $opt(I')>(\log N)^{\frac{1}{\epsilon(k')}}$ can be used to decide if an input graph $G$ has $k$-clique in $f(g(k))n^{O(1)}$ time.


\section{Conclusion}
We have improved the hardness approximation factor for the parameterized set cover problem using a simple reduction. Our result shows that in order to prove inapproximability of parameterized set cover, it suffices to prove the hardness of set cover problem with small universe set. A natural question is:

Is there any algorithm that can, given an $n$-vertex  set cover instance $I$ and an integer $k$, outputs a new instance $I'$ and an integer $k'$ in $f(k)\cdot n^{O(1)}$ time for some computable function $f : \mathbb{N}\to\mathbb{N}$ such that
\begin{itemize}
\item $k'=g(k)$ for some computable function $g : \mathbb{N}\to\mathbb{N}$,
\item $opt(I)\le k$ if and only if $opt(I')\le k'$,
\item $|U(I')|\le h(k)\cdot (\log |S(I')|)^{O(1)}$ for some computable function $h : \mathbb{N}\to \mathbb{N}$.
\end{itemize}

A positive answer to the above question would imply that \textsc{SetCover} parameterized by the optimum solution size has no $(\log n)^{1/\epsilon(k)}$-approximation FPT algorithm assuming $\W 2\neq FPT$. Of course, if we just want a $\rho$-factor hardness of approximation, then it suffices to have $|U(I')|\le h(k) |S(I')|^{O(1/\rho^k)}$.
Note that using  Dynamic Programming, \textsc{SetCover} can be solved  in $2^{|U(I)|}(|U(I)|+|S(I)|)^{O(1)}$ time \cite{cygan2015parameterized}. We do not expect to reduce the size of universe set below $o(k\log n)$ under ETH.

Our hardness result is far from matching the $(1+\ln n)$ approximation ratio of the greedy algorithm in polynomial time. Could it be the case that there exists a $(\ln n)^{1/\rho(k)}$-approximation algorithm for \textsc{SetCover} with running time $n^{k-\epsilon}$? What is the best approximation ratio we can achieve for parameterized set cover in $n^{k-\epsilon}$ time?

\bigskip

\noindent\textbf{Acknowledgement}
This work was supported by JSPS KAKENHI Grant Number JP18H05291.
The author wishes to thank the anonymous referees for their detailed comments.

\bibliography{ref} 
\end{document}